\documentclass[a4paper,10pt]{article}
\usepackage[utf8]{inputenc}

\usepackage[T1]{fontenc}
\usepackage{pifont}
\usepackage{amsmath}
\usepackage{amssymb}
\usepackage{amsthm}
\usepackage{amscd}
\usepackage{bm}

\newtheorem{theorem}{Theorem}[section]
\newtheorem{lemma}[theorem]{Lemma}

\newcommand\numberthis{\addtocounter{equation}{1}\tag{\theequation}}

\newcommand{\R}{\mathbb{R}}
\newcommand{\C}{\mathbb{C}}

\newcommand{\s}{\mathbb{S}}

\newcommand{\abs}[1]{\lvert#1\rvert}
\newcommand{\norm}[1]{\lVert#1\rVert}
\newcommand{\bvec}[1]{\boldsymbol{#1}}

\newcommand{\inprodtwo}[2]{\left \langle #1 , #2\right \rangle}
\newcommand{\innorm}[1]{\langle #1 \rangle}
\newcommand{\dotprod}[2]{
  \bvec{#1}\mkern1mu{\cdot}\mkern1mu\bvec{#2} \,}

\newcommand{\h}{\mathcal{H}}
\newcommand{\hB}{\widetilde{\mathcal{H}}}
\newcommand{\Dom}[1]{\mathcal{D}( #1)}
\newcommand{\Q}[1]{\mathcal{Q}(#1)}

\newcommand{\Pa}{\mathbf{P}_\mathbf{A}}
\newcommand{\Sa}{\mathbf{S}_\mathbf{A}}
\newcommand{\Paux}{\widetilde{\mathbf{P}}_\mathbf{A}}

\newcommand{\sigvec}{\boldsymbol{\sigma}}
\newcommand{\curl}{\operatorname{curl}}

\title{A  criterion for the existence of zero modes for the Pauli operator with fastly decaying fields}
\author{Rafael D. Benguria and Hanne Van Den Bosch 
\\  {\small Pontificia Universidad Cat\'{o}lica de Chile} \\
{\small Av. Vicu\~{n}a Mackenna 4860},
{\small Santiago (Chile)}
}

\begin{document}

\maketitle

\begin{abstract}
 We consider the Pauli operator in $\R^3$ for magnetic fields in $L^{3/2}$ that decay at infinity as $\abs{x}^{-2-\beta}$ with $\beta > 0$.
 In this case we are able to prove that the existence of a zero mode for this operator 
 is equivalent to a quantity $\delta(\bvec B)$, defined below, being equal to zero.
 Complementing a result from \cite{BalinskyEvansLewis}, this implies that for the class of magnetic fields considered,
 Sobolev, Hardy and CLR inequalities hold whenever the magnetic field has no zero mode.
\end{abstract}

\section{Introduction}
Consider the Pauli operator $\Pa$ acting on $L^2(\R^3, \C^2) \equiv \h$,
formally defined by 
\[
\Pa = (\mathbf{p}-\mathbf{A})^2 - \sigvec \cdot \mathbf{B}
\]
where $\mathbf{B} =  \curl \mathbf{A}$.
In appropriate units, this operator describes the kinetic energy of a non-relativistic electron in the magnetic field $\bvec B$.
We will also need the Schrödinger operator $\Sa = (\mathbf{p-A})^2$, which gives the kinetic energy of a spinless particle in a magnetic field.
An element of the kernel of $\Pa$ is called a zero mode for the corresponding Pauli operator.

The importance of zero modes for the Pauli operator was first pointed out in \cite{FrohlichLiebLoss}, 
where the authors realized that their existence would imply a critical value of the nuclear charge $Z$
in order to have a bounded ground state energy for a one-electron atom in a magnetic field.
In \cite{LossYau}, the first examples of magnetic fields producing zero modes were given.
Further examples were given in \cite{AdamMuratoriNash1999, AdamMuratoriNash2000, Elton2000, ErdosSolovej}. 
\cite{AdamMuratoriNash2000} provides explicit examples of magnetic fields with an arbitrary number of zero modes 
while in \cite{Elton2000} a compactly supported magnetic field having a zero mode is constructed.
In \cite{ErdosSolovej} the authors use a geometrical approach which allows, for a certain class of magnetic fields on $\R^3$,
to relate the problem to the one on $\s^2$, which is better understood.

All of the above papers deal with the problem of describing the kernel of the Pauli operator for fixed magnetic fields.
A different point of view is adopted in \cite{BalinskyEvans} and \cite{Elton2002}.
In these cases the authors describe the set of magnetic fields producing zero modes, 
in \cite{BalinskyEvans} for $\bvec B \in L^{3/2}$ and in \cite{Elton2000} for continuous $\bvec A$ decaying as $o(\abs{x}^{-1})$.
Both authors reach the conclusion that magnetic fields on $\R^3$ producing zero modes are rather \emph{rare}
which contrasts heavily with the situation in $\R^2$.

The existence of zero modes for the Pauli operator makes 
it impossible to use the kinetic energy of a wave function to control its potential energy
as it is done for (magnetic) Schrödinger operators by Hardy's inequality or the CLR-bound (\cite{Cwikel, Lieb1980, Rozenbljum}).
However, in \cite{BalinskyEvansLewis} it was shown 
that it is still possible to obtain this type of bounds for certain magnetic fields.
Here, the goal is to give a more precise description of the class of magnetic fields for which this bound holds.
In order to make this statement precise, we first need to review some results of \cite{BalinskyEvans, BalinskyEvansLewis}.

If $\abs{\mathbf{B}} \in L^q$ for some $q \in [\frac{3}{2}, \infty]$, $\Sa$ and $\Pa$ have the same form domain $\Q{\Sa}$.
Both operators can be defined as Friedrich's extensions of the respective quadratic forms.
In addition, we will need the operator $\Paux \equiv \Pa + \abs{\bvec{B}} $, with the same form domain.
Since $\Paux \geq \Sa$, $\mathrm{ker} (\Paux) = \{0\}$, so its range is dense in $\h$.
The auxiliary Hilbert space $\hB$ is defined as the completion of $\Q{\Sa}$ with respect to the norm
\[
\norm{u}_{\hB}^2 = \bigl(u,\Paux u  \bigr).
\]
This space is not a subspace of $\h$.
Its definition ensures $\Paux^{-1/2}$ considered as an operator from 
$\mathrm{Ran}(\Paux^{1/2} )$ to $\hB$ preserves norms.
As previously remarked, its domain is dense in $\h$.
On the other hand, $\mathrm{Ran} ( \Paux^{- 1/2}) = \Dom{ \Pa^{1/2}} = \Q{\Pa} = \Q{\Sa}$, which is dense in $\hB$ by construction.
This means $\Paux^{-1/2}$ can be extended to a unitary operator $U$ from $\h$ to $\hB$.
Multiplication by $\abs{\bvec{B}}^{1/2}$ is a bounded operator from $\hB$ to $\h$.
This allows us to define 
\begin{align*}
 &S = \abs{\bvec{B}}^{1/2} U : \h \rightarrow \h,  \\
 &S = \abs{\bvec{B}}^{1/2} (\Pa + \abs{\bvec{B}})^{-1/2} \quad \text{ on } \mathrm{Ran}(\Paux^{1/2}).
\end{align*}
Finally, define
 \begin{equation} \label{eq : def_delta}
  \delta({\bvec{B}}) = \inf \limits_{\norm{f}=1, Uf \in \h} \norm{(1-S^*S)f}.
 \end{equation}

 With these definitions, we can state the main result.
\begin{theorem} \label{thm : main}
 If $\bvec B \in L^{3/2}$ is such that $\delta({\bvec{B}}) = 0$ 
 and there exists $\beta > 0$, $C \geq 0$ and $r_0 \geq 0$ such that
 \[
 \abs{\bvec B} (x) \leq C \abs{x}^{- 2 - \beta}
 \]
 for all $\abs{x} \geq r_0$,
 then the associated Pauli operator $\Pa$ has a zero mode.
\end{theorem}

We do not know whether the condition on the decay of $\bvec B$ is optimal. 
In any case it can be replaced by the condition on the vector potential $\bvec{A}$ in hypothesis of lemma \ref{lemma : integrability}.
Our method does not work without this additional decay of $\bvec A$.

The quantity $\delta(\bvec B)$ was introduced in \cite{BalinskyEvansLewis} were the following result was proven:
\begin{theorem}[Balinsky, Evans, Lewis, \cite{BalinskyEvansLewis}]
 If $\bvec{B} \in L^{3/2}$, then
 \begin{equation} \label{eq : BEL}
  \Pa \geq \delta(\bvec{B}) \, \Sa.
 \end{equation}

\end{theorem}
If $\delta({\bvec{B}})>0$, this result allows to deduce for instance a Hardy inequality for $\Pa$.
If the Pauli operator corresponding to the magnetic field $\bvec{B}$ has a zero mode, then $\delta({\bvec{B}}) = 0$.
The content of theorem \ref{thm : main} is precisely the converse of this. 
For magnetic fields that decrease sufficiently fast at infinity,
$\delta({\bvec{B}}) = 0$ implies the existence of a zero mode for the corresponding Pauli operator.
Unfortunately, inequality \eqref{eq : BEL} still contains the positive but unknown quantity $\delta({\bvec{B}})$.

The remainder of this paper contains the proof of theorem \ref{thm : main}. 
The next section contains some preliminary lemmas while the third section concludes the proof.

\section{Simplifying the problem}

To prove theorem \ref{thm : main} we will first simplify the statement,
by reducing the condition $\delta(\bvec{B}) = 0$ to a simpler one 
and changing the hypothesis on the decay of $\bvec{B}$ into a hypothesis on $\bvec A$.
This is done in the following two lemmas.

\begin{lemma} \label{lemma : simplified_var_prob}
 If $\delta({\bvec{B}}) = 0$, then
 \begin{equation} \label{eq : simplified_var_prob}
  \inf \limits_{\substack{g \in \Q{\Sa}  \\ (g, \abs{\bvec{B}} g) \neq 0}} \frac{(g,\Pa g)}{(g, \abs{\bvec{B}} g)} = 0.
 \end{equation}
\end{lemma}

\begin{proof}
 First, observe that if
 \[
 \inf \limits_{\norm{f}=1, Uf \in \h} \norm{(1-S^*S)f} = 0,
 \]
 then
 \[
 \sup \limits_{\norm{f}=1, Uf \in \h} \norm{Sf} = 1.
 \]
To see this, first notice that for any $f \in \h$, $\norm{Sf} \leq \norm{f}$, so the $sup$ in the above expression is at most $1$.
Now if $f_n$ is a minimizing sequence for the first problem,
\[
(1- S^* S)f_n \rightarrow 0 \text{ in } L^2
\]
 so in particular
 \[
 \bigl(f_n,(1- S^* S) f_n \bigr) \rightarrow 0.
 \]
 This means $\norm{Sf_n}^2= \bigl(f_n, S^*S f_n\bigr) \to 1$.
 
 Since the range of $\Paux$ is dense in $\h$ and $S$ is bounded, 
 nothing is lost by restricting the $\sup$ to functions $f \in \mathrm{Ran}(\Paux^{1/2})$.
 For these functions the condition $Uf \in \h$ is trivially satisfied. The problem can then be rewritten in terms of $g = Uf$:
 \begin{align*}
 1 =  \sup \limits_{\norm{f}=1, Uf \in \h} \norm{Sf} 
	      &= \sup\limits_{f \in \mathrm{Ran}(\Paux^{1/2})\setminus \{0\}}
			    \frac{ \norm{Sf}}{\norm{f}} \\
	      &= \sup\limits_{g \in \Dom{\Paux^{1/2}}\setminus \{0\}}
			      \frac{ \norm{\abs{\bvec{B}}^{1/2}g} }{ \norm{\Paux^{1/2}g} } \\
 \end{align*}
The result is obtained by expanding ${ \norm{\Paux^{1/2}g} }^2 = (g, \Pa g) + (g, \abs{\bvec{B}} g) $ and using $\Dom{\Paux^{1/2}} = \Q{\Sa}$:
\[
 1 = \sup \limits_{\norm{f}=1, Uf \in \h} \norm{Sf}^2 
	     	      = \sup\limits_{g \in \Q{\Sa}\setminus \{0\}}
			      \left( \frac{(g, \Pa g)}{(g, \abs{\bvec{B}} g)} + 1 \right)^{-1},
\]
which is only possible if
\[
 \inf \limits_{\substack{g \in \Q{\Sa}  \\ (g, \abs{\bvec{B}} g) \neq 0}} \frac{(g,\Pa g)}{(g, \abs{\bvec{B}} g)} = 0. \qedhere
 \]
\end{proof}

Then, we show that the imposed decay of $\bvec B$ implies 
a good decay of $\bvec A$ if we fix the gauge

\begin{equation} \label{eq : def_A}
 \frac{1}{4 \pi}  \bvec A(x) \equiv \int \frac{x-y}{\abs{x-y}^3} \times \bvec B (y) dy .
\end{equation}
Note that $\bvec A$ as defined above is in $L^3$ by the weak Young inequality.

\bigskip

\begin{lemma} \label{lemma : B-decay->A-decay}
 If $\bvec B \in L^{3/2}$ is such that there exists $\beta > 0$, $C_B \geq 0$ and $r_0 \geq 0$ such that
 \[
 \abs{\bvec B} (x) \leq C_B \abs{x}^{- 2 - \beta}
 \]
 for all $\abs{x} \geq r_0$,
 then  there exist $r_1\geq r_0$ and $C_A$ such that
 \[
\abs{\bvec A}(x) \equiv 4 \pi \bigl| \int \frac{x-y}{\abs{x-y}^3} \times \bvec B (y) dy \bigr| \leq C_A \abs{x}^{-1-\alpha}
 \]
for $\alpha = \min(1/2, \beta/2)$ and all $\abs{x} \geq r_1$
\end{lemma}

\begin{proof}
Take $r_1= \max((2 r_0)^{2},1)$.
 Take any $x$ such that $\abs{x}\geq r_1$ and define $r_x = \abs{x}^{1/2}/2 \geq r_0$.
 Split the domain of integration in the definition of $\bvec A$ in two parts and apply Hölder's inequality to the first part to obtain
 \begin{align*}
  \abs{\bvec A} (x)&\leq 4\pi \int_{B_{r_x}} \abs{\bvec B (y)} \abs{x-y}^{-2} dy 
	    + 4\pi \int_{\overline{B_{r_x}}} \abs{\bvec B (y)} \abs{x-y}^{-2} dy \\
	    &\leq 4\pi \norm{\bvec{B}}_{3/2} \left(\int_{B_{r_x}}\abs{x-y}^{-6} dy \right)^{1/3}
	    + 4 \pi C_B \int_{\overline{B_{r_x}}} \abs{y}^{-2-\beta} \abs{x-y}^{-2} dy 
 \end{align*}
The integrand in the first term is bounded, so 
\begin{align*}
 \int_{B_{r_x}}\abs{x-y}^{-6} dy &\leq \frac{4 \pi }{3} r_x^3 (\abs{x}-r_x)^{-6} \\
		  &\leq \frac{2^5 \pi}{3} \abs{x}^{-9/2}
\end{align*}

The second integral requires some more care:
\begin{align*}
 \int_{\overline{B_{r_x}}} \abs{y}^{-2-\beta} \abs{x-y}^{-2} dy 
	& =4 \pi \int_{r_x}^\infty  r^{-\beta}dr \int_{-1}^1 dt (\abs{x}^2 + r^2 - 2 r\abs{x} t)^{-1} \\
	& = 2 \pi \abs{x}^{-1} \int_{r_x}^\infty  r^{-\beta-1}  \ln\left(\frac{\abs{x}+r}{\abs{\abs{x}-r}}\right) dr \\
	& = 2 \pi \abs{x}^{-1-\beta} \int_{r_x/\abs{x}}^\infty  t^{-\beta-1}  \ln\left(\frac{1+t}{\abs{1-t}}\right) dt.
\end{align*}
This last integral is finite since for large $t$, 
the integrand is bounded by a constant times $t^{-\beta - 1}$,
while for $t$ close to $1$ it diverges only as a logarithm.
Separating the range of integration in $r_x / x \leq t \leq 1/2$ and $t > 1/2$
we note that the first part gives a contribution that behaves as $C_1 (r_x/\abs{x})^{-\beta} $ 
while the contribution of the second part can be bounded by a constant.
This means
\begin{align*}
 \int_{\overline{B_{r_x}}} \abs{y}^{-2-\beta} \abs{x-y}^{-2} dy 
 &\leq \abs{x}^{-1-\beta}\left(C_1 \left(\frac{r_x}{\abs{x}}\right)^{-\beta} + C_2 \right)  \\
 &\leq C_1 2^\beta \abs{x}^{-1- \beta/2} + C_2 \abs{x}^{-1-\beta}.
 \end{align*}
We conclude
\[
\abs{\bvec A}(x) \leq C_A (\abs{x}^{-1-1/2}+ \abs{x}^{-1-\beta/2}) \leq 2 C_A \abs{x}^{-1-\alpha}. \qedhere
\]
\end{proof}

\section{Compactness and Integrability}

Now we use a compactness-argument to find a candidate zero mode if the infimum in equation \eqref{eq : simplified_var_prob}
equals zero.

\begin{lemma}
If $\bvec B \in L^ {3/2}$,
and $\delta({\bvec{B}}) = 0$ then there exist $g \in W^{1,2}_{\mathrm{loc}} \cap L^6$ such that
 \[
 \sigvec \cdot (\bvec p - \bvec A) g = 0
 \]
 in the particular gauge for $\bvec A$ defined in \eqref{eq : def_A}.
\end{lemma}
\begin{proof}
 Take $(g_n)$ a minimizing sequence for the problem \eqref{eq : simplified_var_prob} 
 with $(g_n, \abs{\bvec{B}} g_n) = 1$.
 Then $(g_n, \Pa g_n)$ is bounded, which implies by the diamagnetic inequality that $(p g_n)$ is bounded in $L^2$
 so $(g_n)$ is bounded in $L^6$.
   By the Banach-Alaoglu theorem, 
   this guarantees the existence of a subsequence such that $p g_n$ converges weakly in $L^2$ to some $p g$
   and $ g_n \rightharpoonup  g$ weakly in $L^6$.
 Since $\abs{\bvec{B}} \in L^{3/2}$, this implies $(g, \abs{\bvec{B}} g) = 1$, so $g \neq 0$.
 In addition, since $\bvec{A} \in L^3$,
 $(\bvec{A} g_n)$ is bounded in $L^2$ so we can assume $\bvec A g_n \rightharpoonup \bvec A g$ weakly in $L^2$.
 Using the fact that $L^p$-norms are weakly lower-semi-continuous, we obtain $ \norm{\dotprod{\sigma}{(p-A)}g}_2 = 0$.
\end{proof}

To conclude the proof of theorem \ref{thm : main} we only need to show that this candidate zero mode is in $L^2$.
This is achieved by using the decay of $\bvec{A}$ given by lemma \ref{lemma : B-decay->A-decay} in a bootstrap argument.
The procedure is not that straightforward since the decay of $\bvec A g$ 
and the Pauli equation imply only a decay of $\dotprod{\sigma}{p} g$, which does not directly imply the decay of $\bvec p g$.

\begin{lemma} \label{lemma : integrability}
  If there exist $\alpha > 0$ and $r_1 > 0$ such that
  $\abs {\bvec A}(x) < C_A \abs{x}^{-1-\alpha}$ 
  for all $x \in \R^3$ with $\abs{x} \geq r_1$
  and $g \in W^{1,2}_{\mathrm{loc}} \cap L^p$, with $p \geq 2$, is such that
 \[
 \sigvec \cdot (\bvec p - \bvec A) g = 0,
 \]
 then $g \in L^{2}$.
\end{lemma}

In order to prove this lemma, one more technical lemma will be necessary.
Its proof can be found in the appendix.
The inner product in $L^2(\s^2, \C^2)$ will be denoted by $\inprodtwo{\cdot}{\cdot}$. 
When $f$ and $g$ are defined on all of $\R^3$, we will abuse notation and write 
$\inprodtwo{f}{g}(r) \equiv \inprodtwo{f(r \bvec \omega)}{g(r \bvec \omega)}$.
We will also use the notation $\innorm{f}(r) = \inprodtwo{f}{f}^{1/2}(r)$.

\begin{lemma} \label{lemma : regularity}
If $f \in W^{1,2}_{loc}(\R^3)$ then $\innorm{f} \in W^{1,2}([a,b])$ for all $b > a > 0$,
and its weak derivative equals
\[
h(r) = \left\{ 
 		\begin{array}{ll}
              \innorm{f}^{-1}(r) \Re \inprodtwo{f}{\partial_r f}  & \text{ if } \innorm{f}(r) > 0\\
               0 & \text{ else.}
              \end{array}
              \right. 
\]
In particular $\innorm{f}$ is continuous except maybe at $0$. 
\end{lemma}

\begin{proof}[Proof of lemma \ref{lemma : integrability}]

Define 
\[
K = - 1 - \dotprod{\sigma}{L}.
\]
which can be considered as a self-adjoint operator on $L^2(\s^2,\C^2)$ with eigenvalues $\pm 1, \pm 2, \dots $ 
(see for instance \cite{Johnson}, section 1.5).
Write $g = g_+ + g_-$ where $\inprodtwo{g_+}{Kg_+} > 0$ and $\inprodtwo{g_-}{Kg_-} < 0$. 
If $g \in L^p(\R^3)$, there exists $C > 0$ such that
\[
\int_{\s^2} \abs{g}^p (r \omega) d \omega \leq C r^{-3}.
\]
By Jensen's inequality, this implies
\begin{align*}
C r^{-3} \geq \int_{\s^2} \abs{g}^p (r \omega) d \omega 
	&\geq (4 \pi)^{1-p/2} \left( \int_{\s^2} \abs{g}^2 (r \omega) d \omega \right)^{p/2} \\
	&=  (4 \pi)^{1-p/2} \bigl( \inprodtwo{g_+}{g_+} + \inprodtwo{g_-}{g_-} \bigr)^{p/2},
\end{align*}
so both $\innorm{g_+}(r)$ and $\innorm{g_-}(r)$ decay as $ C r^{-3/p}$.

At first, we will prove the theorem in the case that $g_+$ and $g_-$ are $C^2$\nobreakdash-\hspace{0pt}functions.
The Pauli operator can be written conveniently as
\[
\dotprod{\sigma}{p} = (\dotprod{\sigma}{\hat x})^2\dotprod{\sigma}{p} 
= - i \dotprod{\sigma}{\hat x} \left(\partial_r + \frac{K+1}{r}\right),
\]
where the operator inside the parenthesis commutes with $K$.
This allows to rewrite the equation for $g$ as
\[
\partial_r g + \frac{K+1}{r} g = i \dotprod{\sigma}{\hat x} \dotprod{\sigma}{A} g.
\]
For shortness, define $\sigma_A = \dotprod{\sigma}{\hat x} \dotprod{\sigma}{A}$.
The only property of this matrix needed is $\norm{\sigma_A (r \omega)} \leq C_A r^{-1-\alpha}$ when $r \geq r_1$.
Taking the $\C^2$ product with $g_+$ and $g_-$ and integrating over $\s^2$, we obtain
\begin{align} \label{eq : int_S2}
 \inprodtwo{g_+}{\partial_r g_+}(r) & = -\frac{1}{r}\inprodtwo{g_+}{(K+1)g_+}(r) + i \inprodtwo{g_+}{\sigma_A (g_+ + g_-)} \\
  \inprodtwo{g_-}{\partial_r g_-}(r) & = -\frac{1}{r} \inprodtwo{g_-}{(K+1)g_-}(r) + i \inprodtwo{g_-}{\sigma_A (g_+ + g_-)}. \nonumber
\end{align}

By taking the real part of these equations, we obtain a differential equation for $\innorm{g_+}$ and $\innorm{g_-}$:

\begin{align*}
 \frac{d}{dr} \innorm{g_+}^2 & = -2 \frac{1}{r}\inprodtwo{g_+}{(K+1)g_+}(r) - 2 \Im \inprodtwo{g_+}{\sigma_A (g_+ + g_-)} \\
\frac{d}{dr} \innorm{g_-}^2 & = -2\frac{1}{r} \inprodtwo{g_-}{(K+1)g_-}(r) - 2 \Im \inprodtwo{g_-}{\sigma_A (g_+ + g_-)}.
\end{align*}

Defining $ \bar g_+ = g_+ r^{2}$ we get the system of equations
\begin{align*}
 \frac{d}{dr} \innorm{\bar g_+}^2 & = -2 \frac{1}{r}\inprodtwo{\bar g_+}{(K-1)\bar g_+}(r) 
	      - 2 \Im \inprodtwo{\bar g_+}{\sigma_A (\bar g_+ r^2 g_-)} \\
\frac{d}{dr} \innorm{ g_-}^2 & = -2\frac{1}{r} \inprodtwo{g_-}{(K+1)g_-}(r) - 2 \Im \inprodtwo{g_-}{\sigma_A ( r^{-2}\bar g_+ + g_-)}.
\end{align*}

Fix $r \geq r_1$.
We now use a bootstrap argument to obtain $\innorm{g_\pm}(r) \leq C r^{-2}$.
As remarked previously $\innorm{g_-}^2(r) \leq C r^{-\epsilon}$ and $\innorm{\bar g_+}^2(r) \leq C r^{4-\epsilon}$ with $\epsilon = 3/p$. 
We will see the equations imply $\innorm{ g_-}^2(r) \leq C' r^{-\epsilon-\alpha}$ 
and $\innorm{\bar g_+}^2(r) \leq C' r^{4-\epsilon_1}$ where $\epsilon_1 = \min (\epsilon + \alpha, 4)$.

For $\bar g_+$, we can use $\inprodtwo{\bar g_+}{K \bar g_+} \geq \innorm{\bar g_+}^2$ in order to obtain
\begin{align*}
 \innorm{\bar g_+}^2(r) &= \int_{r_1}^r -2 s^{-1}\inprodtwo{\bar g_+}{(K-1)\bar g_+}(s) 
                         - 2 \Im \inprodtwo{\bar g_+}{\sigma_A (\bar g_+ + s^2 g_-)} ds + C_1                       \numberthis \label{eq : integ_f} \\
	& \leq 2 \int_{r_1}^r \abs{\inprodtwo{\bar g_+}{\sigma_A \bar g_+}(s) } 
	                  + s^2 \abs{\inprodtwo{\bar g_+}{\sigma_A g_-}(s) } d s + C_1 \\
	& \leq 4 C C_A \int_{r_1}^r s^{4-\epsilon -1-\alpha} + C_1 \\
	& = \frac{4 C C_A}{4 - \epsilon - \alpha}(r^{4 - \epsilon - \alpha}-1) + C_1.
\end{align*}
For $g_-$, we can use the fact $\innorm{g_-}$ tends to zero as $r \to \infty$ and  $\inprodtwo{g_-}{K g_-} \leq - \innorm{g_-}^2$ to write
\begin{align*}
 \innorm{g_-}^2(r)& = \int_r^\infty + 2 s^{-1}\inprodtwo{g_-}{(K+1)g_-}(s) 
			    + 2 \Im \inprodtwo{g_-}{\sigma_A (s^{-2}\bar g_+ +  g_-)} ds  \numberthis \label{eq : integ_h}\\
	& \leq \int_r^\infty  2 \abs{ \inprodtwo{g_-}{\sigma_A (s^{-2}\bar g_+ +  g_-)}} ds  \\
	& \leq 2 C C_A \int_r^\infty  2 s^{- \epsilon - 1 - \alpha} ds \\
	& = \frac{2 C C_A}{\epsilon + \alpha} r^{- \epsilon - \alpha}.
\end{align*}
By iterating this procedure a finite number of times we reach the conclusion $\innorm{g_+}(r) \leq C r^{-2}$ and $\innorm{g_-}(r) \leq C r^{-2}$,
so $g \in L^2(\R^3)$.
This concludes the proof of the lemma when $g_+$ and $g_-$ are $C^2$-functions.

In the general case, $g$ has a decomposition in a series of spherical spinors (see for example \cite{Johnson}, section 1.5)
where the coefficients are functions of $r$ belonging to $W^{1,2}_{loc}(\R_+, r^2 dr)$.
By taking the projections on the positive and negative eigenspaces of $K$ and using dominated convergence, 
we conclude $g_+$ and $g_-$ are in $W^{1,2}_{loc}(\R^3)$. 
Thus, by Fubini's theorem, $g_\pm$, and $\partial_r g_\pm $ are in $L^2(\s^2(r))$ for almost every $r > 0$.
This justifies the integration over $\s^2$ used to obtain \eqref{eq : int_S2}.

By lemma \ref{lemma : regularity}, $\innorm{g_+}$ and $\innorm{g_-}$ are in $ W^{1,2}([a,b])$ for any $b > a > 0$ and thus continuous.
The use of the fundamental theorem of calculus in 
\eqref{eq : integ_f} can be justified by applying it to a sequence of $C^\infty$-functions converging to $\innorm {g_+}$
pointwise and in $ W^{1,2}([r_1,r])$. 
In the same way we can obtain $\innorm {g_-}^2(r) = -\int_{r}^{r_2} \frac{d}{dr} \innorm{g_-}^2(r) dr  + \innorm {g_-}^2(r_2)$
for any $r_2 > r > r_1$.
Since $\frac{d}{dr} \innorm{g_-}^2(r)$  is in $L^1([r, +\infty))$, we can let $r_2 \to \infty$ in order to obtain \eqref{eq : integ_h}.

\end{proof}

\section*{Acknowledgements}

Work partially supported by Fondecyt (Chile) project 112--0836 
and the Iniciativa Cient\'\i fica Milenio (Chile) through the Millenium Nucleus RC--120002 ``F\'\i sica Matem\'atica'' (R.D.B and H.VDB) 
and partially by Conicyt (Chile) through CONICYT-PCHA/Doctorado Nacional/2014 
and Beca Ayudante de VRI (H.VDB).

\section*{Appendix A: proof of lemma \ref{lemma : regularity}}

 By Fubini's theorem, $f$ and $\partial_r f$ are in $L^2(\s^2(r))$ for almost every $r > 0$ and 
 $\innorm{f}$, $\innorm{\partial_r f}$ are in $L^2_{loc}(\R_+, r^2 dr)$.
 Fix $b > a > 0$ and define the annulus $A = \{ x \in \R^3 |  a \leq \abs x \leq b \}$.
 
 Fix $\epsilon > 0$.
 As a first step, we will prove $f_\epsilon  \equiv (\innorm{f}^2 +\epsilon)^{1/2} \in   W^{1,2}([a,b])$. 
 Define $h_{\epsilon} = \innorm{f_\epsilon}^{-1} \Re \inprodtwo{f}{\partial_r f}$.
 By the Cauchy-Schwarz inequality $h_\epsilon $ is in $L^2([a, b])$.
 It remains to check whether $h_\epsilon $ is the distributional derivative of $f_\epsilon$.
 
 To this end, take a sequence $(f_n) \subset C^1(A, \C^2)$ 
 approaching $f$ in $W^{1,2}(A)$ and pointwise almost everywhere in $A$.
 This means $\innorm{f_n} \to \innorm{f}$ in $L^2([a,b])$
 so by extracting a subsequence we may assume $\innorm{f_n}(r) \to \innorm{f}(r)$ for almost every $r \in [a,b]$.
 Define $h_n  \equiv \partial_r (\innorm{f_n}^2+ \epsilon) ^{1/2}$. We have 
 $h_n = (\innorm{f_n}^2+ \epsilon)^{-1/2} \Re \inprodtwo{f_n}{\partial_r f_n}(r)$.
 In order to conclude, we should prove that, for any test function $\phi \in C_0^\infty ([a,b])$,
 \[
 \int_{a}^{b} \phi(r) h_n(r) dr 
      \to 
	    \int_{a}^{b} \phi(r) h_{\epsilon}(r) dr \quad \text{ as } n \to \infty.
 \]
 
To achieve this, fix $\phi \in C_0^\infty ([a,b])$ and define $\Phi_n =(\innorm{f_n}^2+ \epsilon)^{-1/2} \phi f_n $ 
and $\Phi_{\epsilon} = f_\epsilon^{-1} \phi f $.
$\Phi_n(x)$ converges to $\Phi_{\epsilon}(x)$ 
when $x \in A$ is such that $\innorm{f_n}(\abs x)$ converges to $\innorm{f}(\abs x)$ and $f_n(x) \to f(x)$,
which holds for almost every $x$ in $A$.
Since $\bigl(\Phi_n \bigr)$ is bounded in $L^2(A)$, by dominated convergence $\Phi_n \to \Phi_\epsilon$ in $L^2(A[a,b])$.
This allows us to obtain
\begin{align*}
 \int_{a}^{b}  \abs{\phi (h_n - h_\epsilon) }
	& \leq \int_{a}^{b} \abs{\inprodtwo{\Phi_n}{\partial_r f_n}- \inprodtwo{\Phi_\epsilon}{\partial_r f}} \\
	& \leq \int_{a}^{b} \abs{\inprodtwo{\Phi_n-\Phi_\epsilon}{\partial_r f_n}}
	       +\int_{a}^{b} \abs{\inprodtwo{\Phi_\epsilon}{\partial_r f_n-\partial_r f}} \\
	& \leq a^{-2} \norm{\Phi_n - \Phi_\epsilon}_{2,A} \norm{\partial_r f_n}_{2,A} 
	   + a^{-2} \norm{\Phi_\epsilon}_{2,A} \norm{\partial_r f_n- \partial_r f}_{2,A}.
 \end{align*}
In the last line, we used $1 \leq a^{-2} r^2 $ in the domain of integration to transform the integral over an interval in an integral over $A$.
Since $f_n$ tends to $f$ in $W^{1,2}(A)$, the second term tends to zero and the second factor of the first term is bounded.
As previously remarked, $\Phi_n - \Phi_\epsilon$ tends to zero in $L^2(A)$ so the first term goes to zero too.
This means $f_\epsilon \in W^{1,2}([a,b])$ and its distributional derivative equals $h_\epsilon$.

Now, we can let $\epsilon$ tend to zero. 
Then $f_\epsilon \to \innorm{f}$  and
$h_\epsilon (r) \to h(r)$ in $L^2([a,b])$.
We conclude $\innorm{f} \in W^{1,2}([a,b])$ and $h = \frac{d}{dr} \innorm{f}$.
\qed

\bibliography{references2} 

\providecommand{\bysame}{\leavevmode\hbox to3em{\hrulefill}\thinspace}
\providecommand{\MR}{\relax\ifhmode\unskip\space\fi MR }
\providecommand{\MRhref}[2]{%
  \href{http://www.ams.org/mathscinet-getitem?mr=#1}{#2}
}
\providecommand{\href}[2]{#2}
\begin{thebibliography}{10}

\bibitem{AdamMuratoriNash1999}
C.~Adam, B.~Muratori, and C.~Nash, \emph{Zero modes of the {D}irac operator in
  three dimensions}, Phys. Rev. D \textbf{60} (1999), 125001.

\bibitem{AdamMuratoriNash2000}
C.~Adam, B.~Muratori, and C.~Nash, \emph{Degeneracy of zero modes of the
  {D}irac operator in three dimensions}, Phys. Lett. B \textbf{485} (2000), 314
  -- 318.

\bibitem{BalinskyEvans}
A.A. Balinsky and W.D. Evans, \emph{On the zero modes of {P}auli operators}, J.
  Funct. Anal. \textbf{179} (2001), 120 -- 135.

\bibitem{BalinskyEvansLewis}
A.A. Balinsky, W.D. Evans, and R.~T. Lewis, \emph{{S}obolev, {H}ardy and {CLR}
  inequalities associated with {P}auli operators in $\mathbb{R}^3$}, J. Phys. A
  \textbf{34} (2001), L19--L23.

\bibitem{Cwikel}
M.~Cwikel, \emph{Weak type estimates for singular values and the number of
  bound states of {S}chr\"odinger operators}, Ann. of Math. (2) \textbf{106}
  (1977), 93--100.

\bibitem{Elton2000}
D.~M. Elton, \emph{New examples of zero modes}, J. Phys. A \textbf{33} (2000),
  7297--7303.

\bibitem{Elton2002}
D.~M. Elton, \emph{The local structure of zero mode producing magnetic
  potentials}, Comm. Math. Phys. \textbf{229} (2002), 121--139.

\bibitem{ErdosSolovej}
L{\'a}szl{\'o} Erd{\" o}s and Jan~Philip Solovej, \emph{The kernel of {D}irac
  operators on $\mathbb{S}^3$ and $\mathbb{R}^3$}, Rev. Math. Phys. \textbf{13}
  (2001), 1247--1280.

\bibitem{FrohlichLiebLoss}
J.~Fröhlich, E.~H. Lieb, and M.~Loss, \emph{Stability of {C}oulomb systems
  with magnetic fields. {I}. the one-electron atom}, Comm. Math. Phys.
  \textbf{104} (1986), 251--270.

\bibitem{Johnson}
W.~R. Johnson, \emph{Atomic structure theory}, Springer, Berlin, Heidelberg,
  New York, 2007.

\bibitem{Lieb1980}
E.~H. Lieb, \emph{The number of bound states of one-body {S}chroedinger
  operators and the {W}eyl problem}, Proc. Sympos. Pure Math., XXXVI, 1980,
  pp.~241--252.

\bibitem{LossYau}
M.~Loss and H.-T. Yau, \emph{Stabilty of {C}oulomb systems with magnetic
  fields. {III}. zero energy bound states of the {P}auli operator}, Comm. Math.
  Phys. \textbf{104} (1986), 283--290.

\bibitem{Rozenbljum}
G.~V. Rozenbljum, \emph{Distribution of the discrete spectrum of singular
  differential operators}, Izv. Vys\v s. U\v cebn. Zaved. Matematika (1976),
  no.~1(164), 75--86, {\it English transl.} Soviet Math. (iz. VUZ) {\bf 45}
  63-71 (1976).

\end{thebibliography}
\bibliographystyle{amsplain}

\end{document}